\documentclass[11pt]{article}
\usepackage{amsmath}
\usepackage{amssymb}
\usepackage{amsthm}
\usepackage{graphicx}
\usepackage{mathtools}
\usepackage[colorlinks=true,linkcolor=blue, citecolor=purple, pagebackref=true]{hyperref}
\usepackage{enumitem}
\usepackage[capitalise]{cleveref}
\usepackage{tikz}
\allowdisplaybreaks

\usepackage[margin=1in]{geometry}
\usepackage{todonotes}

\renewcommand{\le}{\leqslant}

\renewcommand{\ge}{\geqslant}

\renewcommand{\Pr}{\text{Pr}}
\newcommand{\N}{\mathbb{N}}
\newcommand{\F}{\mathbb{F}}
\newcommand{\Z}{\mathbb{Z}}

\newcommand{\wt}{\text{wt}}

\newcommand{\abs}[1]{\left|#1\right|}

\newcommand{\mb}[1]{\mathbf{#1}}
\newcommand{\bs}[1]{\boldsymbol{#1}}
\newcommand{\mc}[1]{\mathcal{#1}}

\renewcommand{\Pr}[1]{\textbf{Pr}\left[#1\right]}
\newcommand{\condPr}[2]{\textbf{Pr}\left[#1 \; | \; #2\right]}
\newcommand{\Ex}[1]{\textbf{E}\left[#1\right]}

\newtheorem{theorem}{Theorem}[section]
\newtheorem{prop}[theorem]{Proposition}
\newtheorem{lemma}[theorem]{Lemma}
\newtheorem{defn}{Definition}[section]
\newtheorem{remark}{Remark}[section]
\begin{document}
    
\title{Near-Tight Bounds for 3-Query Locally Correctable Binary \\ Linear Codes via Rainbow Cycles}
\author{Omar Alrabiah\thanks{Department of Electrical Engineering and Computer Science, UC Berkeley, Berkeley, CA, 94709, USA. Email: \url{oalrabiah@berkeley.edu}. Research supported in part by a Saudi Arabian Cultural Mission (SACM) Scholarship, NSF CCF-2210823 and V.\ Guruswami's Simons Investigator Award.} 
\and Venkatesan Guruswami\thanks{Department of Electrical Engineering and Computer Science, Department of Mathematics, and the Simons Institute for the Theory of Computing, UC Berkeley, Berkeley, CA, 94709, USA. Email: \url{venkatg@berkeley.edu}. Research supported by a Simons Investigator Award and NSF grants CCF-2210823 and CCF-2228287.}
}
    
\date{}
\maketitle
\thispagestyle{empty}

\begin{abstract}
We prove that a binary linear code of block length $n$ that is locally correctable with $3$ queries against a fraction $\delta > 0$ of adversarial errors must have dimension at most $O_{\delta}(\log^2 n \cdot \log \log n)$. This is almost tight in view of quadratic Reed-Muller codes being a $3$-query locally correctable code (LCC) with dimension $\Theta(\log^2 n)$. Our result improves, for the binary field case, the $O_{\delta}(\log^8 n)$ bound obtained in the recent breakthrough of \cite{KM23} (and the more recent improvement to $O_{\delta}(\log^4 n)$ for binary linear codes announced in \cite{Yan24}).

\smallskip
Previous bounds for $3$-query linear LCCs proceed by constructing a $2$-query locally decodable code (LDC) from the $3$-query linear LCC/LDC and applying the strong bounds known for the former.
Our approach is more direct and proceeds by bounding the covering radius of the dual code, borrowing inspiration from~\cite{IS20}. That is, we show that if $x \mapsto (v_1 \cdot x, v_2 \cdot x, \ldots, v_n \cdot x)$ is an arbitrary encoding map $\F_2^k \to \F_2^n$ for the $3$-query LCC, then all vectors in $\F_2^k$ can be written as a $\widetilde{O}_{\delta}(\log n)$-sparse linear combination of the $v_i$'s, which immediately implies $k \le \widetilde{O}_{\delta}((\log n)^2)$. The proof of this fact proceeds by iteratively reducing the size of any arbitrary linear combination of at least $\widetilde{\Omega}_{\delta}(\log n)$ of the $v_i$'s. We achieve this using the recent breakthrough result of~\cite{ABSZZ23} on the existence of rainbow cycles in properly edge-colored graphs, applied to graphs capturing the linear dependencies underlying the local correction property.
\end{abstract}

\section{Introduction}
\label{sec:intro}

Local correction refers to the notion of correcting a single bit of a received codeword by querying very few other bits of the codeword at random. More concretely, a binary code, which is simply a subset $C \subseteq \{0,1\}^n$, is said to be \emph{locally correctable} using $r \in \N$ queries from a fraction $\delta \in (0,1)$ of errors, abbreviated $(r,\delta)$-LCC, if it can recover any given bit of a codeword $c \in C$ with probability noticeably higher than $1/2$ (say $2/3$) by randomly reading $r$ bits of a received codeword $y \in \{0,1\}^n$ that is at most $\delta n$ away from $c$ in Hamming distance. Usually, we are interested in the case when $\delta$ is a fixed constant bounded away from $0$ as the code length $n \to \infty$, and in this case, we refer to such a code as simply a $r$-LCC.

Throughout this paper, we will restrict our attention to only binary linear codes, particularly binary linear $r$-LCCs. A binary linear code $C$ of block length $n$ is simply a subspace of $\F_2^n$, where $\F_2$ is the field of two elements. If the dimension of $C$ as a $\F_2$-subspace is $k$, then one refers to it as an $[n,k]$ code. A generator matrix of $C$ is an $n \times k$ matrix whose columns form a basis of $C$. Let us fix one such choice of generator matrix $M$, and denote its rows by $v_1,v_2,\dots,v_n \in \F_2^k$. We then have the encoding map $ M : \F_2^k \to C$ given by $Mx = [v_1 \cdot x, v_2 \cdot x, \ldots, v_n \cdot x]^\top$.

Among its many uses, locally correctable codes play a central role in PCP constructions, where they allow to self-correct a function, purportedly a codeword, after a codeword test ascertains that the function is close to a codeword. They thus allow effective noise-free oracle access to a noisy function, with a small price in the number of queries. We refer the reader to the surveys~\cite{Tre04, Yek12, Gop18} for more on the applications and connections of locally correctable codes.

Despite its slew of uses, the best known $r$-LCCs (even existentially) have $n \approx \exp(k^{1/(r-1)})$, which is achieved by the degree $(r-1)$ Reed-Muller code (evaluations of polynomials of degree $(r-1)$ in $m=O_q(\log n)$ variables at all points in $\F_q^m$).\footnote{This code requires $q \ge r+1$, but one can also get say binary codes by picking $q$ to be a power of $2$ and concatenating the Reed-Muller code over $\F_q$ with the binary Hadamard code.} This has remained the case for constant-query local correction since their conception. Indeed, much of the progress on locally correctable codes for a constant number of queries has focused on proving their limitations, specifically for concrete values of $r$.\footnote{This statement holds only for the classical constant query regime. Indeed, there have been some great works for when the number of queries $r$ grows with $n$~\cite{GKS13, KSY14, HOW15, KMRS17, GKORS18} and for relaxed notions of local corrections~\cite{GL19, GRR20, AS21, CY22, KM23b, CY23}. There is also a brighter landscape of lower bounds for harsher error models~\cite{OPC15, BBGKZ20, BGGZ21, BCGLZZ22, BBCGLZZ23, Gup23}.} For $r=1$, it has long been known that $1$-LCCs do not exist~\cite{KT00}. For $r=2$, it has also long been known that one must indeed have $n \ge \exp(\Omega_q(k))$~\cite{GKST06,KdW04}, so the Hadamard code (and the degree one Reed-Muller code) is indeed optimal. 

For $r=3$ and larger, our understanding of $r$-LCCs is abysmal. The known limitations of $r$-LCCs, which also apply to $r$-query locally \emph{decodable} codes (which offer the weaker guarantee of local correction only for the $k$ message symbols encoded by the codeword), stood at the bound $k \le \widetilde{O}(n^{1 - 1/\lceil 2/r \rceil})$~\cite{KdW04, Woo07, Woo12} for a long while. In particular, for $3$-LCCs, the quadratic bound $k \le O(\sqrt{n})$ stood for more than a decade. This was recently improved to $k \le \tilde{O}(\sqrt[3]{n})$ in \cite{AGKM23} (with recent logarithmic factor improvements by~\cite{HKMMS24}), and this bound also applied to $3$-query locally decodable codes (LDCs). Then, in a tour de force breakthrough, Kothari and Manohar~\cite{KM23} gave an exponential improvement and showed that $k \le O_q(\log^8 n)$ for $3$-query \emph{linear} LCCs (over any field $\F_q$). Since there are beautiful constructions of $3$-query linear LDCs of block length sub-exponential in $k$~\cite{Yek08, Rag07, Efr12, DGY11}, their bound demonstrated a strong separation between local decodability and local correctability with $3$ queries for linear codes. Nonetheless, their result left open the optimality of degree $2$ Reed-Muller codes as binary linear $3$-LCCs, which have dimension $k = \Theta(\log^2{n})$. Our main result is that they are (almost) optimal.

\begin{theorem}[Main]
\label{thm:main}
If $C$ is an $[n,k]$ binary linear $(3,\delta)$-LCC, then $k \le O(\delta^{-2} \log^2 n \cdot \log \log n)$.
\end{theorem}

Modulo the $\log\log n$ factor, this settles the dimension versus block length trade-off of 3-query \emph{binary linear} LCCs. Recently, following \cite{KM23}, an improved upper bound of $k \le O(\log^4 n)$ was obtained for binary linear $3$-LCCs in \cite{Yan24}. Even more recently, an independent result of~\cite{KM24} shows an optimal $k \le O(\log^2{n})$ bound for binary linear \emph{design} $3$-LCCs. Such $3$-LCCs have the additional property that the linear dependencies of length $4$ formed by the query sets (see~\cref{def:lcc-combi}) cover each pair of indices in $[n]$ exactly once. We note that a weaker bound of $k \le O(\log^3{n})$ for binary linear design $3$-LCCs was previously shown in~\cite{Yan24}.

Our proof method additionally sheds some light on the structure of binary linear $3$-LCCs. Namely, we prove~\cref{thm:main} by upper bounding the \emph{covering radius of the dual code}.\footnote{The covering radius of a linear code $C_0 \subseteq \F_2^n$ is the minimum $r$ such that every point in $\F_2^n$ is within Hamming distance $r$ from some codeword $c \in C_0$. If $H \in \F_2^{m \times n}$ is a parity check matrix of a linear code $C_0$, then it is the minimum $r$ for which every $s \in \F_2^m$ is the sum of at most $r$ columns of $H$.} This offers a more direct understanding of the structure and limitations of binary linear $3$-LCCs, which can be harder to discern from recent developments~\cite{AGKM23, KM23, HKMMS24, Yan24}. Indeed, all such works proceed by constructing a much longer 2-query LDC from the $3$-query locally correctable linear code and appealing to the known exponential lower bounds for $2$-LDCs~\cite{GKST06, KdW04}.\footnote{See Appendix B of~\cite{AGKM23} and Section 7.6 of \cite{KM23} for the proper formulation of their blocklength lower bound proofs as reductions to $2$-query LDCs.}

Our main result on the covering radius of the dual code of a binary linear $3$-LCC is the following.

\begin{theorem}
\label{thm:intro-cov-rad}
Let $C$ be a binary linear $(3, \delta)$-LCC with generator matrix $M \in \F_2^{n \times k}$. Then every $x \in \F_2^k$ can be expressed as the sum of at most $O(\delta^{-2} \log n \cdot \log \log n)$ rows of $M$. 
\end{theorem}

Since a generator matrix of $C$ is also a parity check matrix of $C^\perp$, Theorem~\ref{thm:intro-cov-rad} as stated upper bounds the covering radius of $C^\perp$. Note that Theorem~\ref{thm:intro-cov-rad} immediately implies Theorem~\ref{thm:main}, as it shows $2^k \le \sum_{j=0}^T {n \choose T} \le n^{T+1}$ for $T = O(\delta^{-2} \log n\cdot  \log \log n)$. We remark here that the degree $2$ Reed-Muller code has a covering radius of $\Theta(\log n)$, which makes our bound in \cref{thm:intro-cov-rad} only a $\log{\log{n}}$ factor away from the optimal bound.

Our inspiration for Theorem~\ref{thm:intro-cov-rad} came from a work of Iceland and Samorodnitsky~\cite{IS20}, who prove that the dual $C^\perp$ of a binary linear $(2,\delta)$-LCC $C$ has $O(\delta^{-1})$ covering radius (which then immediately implies that $|C| \le n^{O(\delta^{-1})}$).\footnote{They also deduce a covering radius upper bound of $O(n^{(r-2)/(r-1)})$ for the $r$-query case by reducing to the $2$-query case. Note that, for $r \ge 3$, the resulting bounds for LCCs are weaker than the best-known ones.} They prove this via analysis of  the ``discrete Ricci curvature" of the ``coset leader graph" associated with $C$. We develop a more elementary treatment of their ideas and give a similar coupling argument to bound the diameter of the Cayley graph $\text{Cay}(\F_2^k, \{v_1,v_2,\dots,v_n\})$, which is isomorphic to their coset leader graph. Note that this diameter is precisely the covering radius of $C^\perp$.  Using our viewpoint, we produce a new proof of the previously known $k \le O(\log n)$ upper bound for linear $2$-query LDCs over any finite field (the proof in \cite{IS20} only applied to LCCs); we present this proof in Appendix~\ref{sec:new-2-ldc-proof}.

\paragraph{Rainbow cycles in properly edge-colored graphs.}
Our proof of Theorem~\ref{thm:intro-cov-rad} crucially relies on finding \emph{rainbow cycles} in properly edge-colored graphs. Rainbow cycles are simply cycles where each color appears at most once. There has been numerous works to that end~\cite{KMSV07, DLS12, Jan20, JS22, Tom22, KLLT22, ABSZZ23}, culminating in the recent breakthrough of~\cite{ABSZZ23} showing that any properly edge-colored $n$-vertex graph with average degree $\Omega(\log{n}\cdot \log\log{n})$ must have a rainbow cycle. This bound is tight up to the $O(\log\log{n})$ factor---if one colors the edges of the Boolean hypercube with their respective direction, then one obtains a properly edge-colored $\log n$-regular $n$-vertex graph that has no rainbow cycles. 

Our $O(\log n \cdot \log \log n)$ bound in our Theorem~\ref{thm:intro-cov-rad} is inherited in a black-box fashion from the rainbow cycle bound of~\cite{ABSZZ23}. Should a tight $\Theta(\log n)$ be established for the minimum average degree guaranteeing a rainbow cycle, we would immediately get an asymptotically tight $O(\log^2 n)$ dimension upper bound for binary linear $3$-LCCs in Theorem~\ref{thm:main}. In fact, in our application, the concerned edge-colored graphs have the further property that each color class has $\Omega(n)$ edges. So it would suffice to improve the rainbow cycle bound for such graphs.

\paragraph{LCC lower bounds from rainbow LDC lower bounds.} Our $3$-LCC result based on rainbow cycles turns out to be a specific instance of a more general reduction from lower bounds for $r$-LCCs to a ``rainbow" form of lower bounds for binary linear $(r-1)$-query LDCs---a stronger form of LDC lower bounds than usual binary linear $(r-1)$-LDC lower bounds. Our main result is the $r=3$ case of this phenomenon, where we have such strong ``rainbow" bounds for binary linear $2$-query LDCs. 

As for bounds on the so-called ``rainbow" binary linear $r$-LDC lower bounds problem, one can prove the same bound of $k \le \widetilde{O}(n^{1-2/r})$ for even $r \ge 4$ known for usual $r$-LDCs in nearly the same fashion! As it turns out, the direct sum transformation of~\cite{KdW04} from $r$-LDCs to $2$-LDCs has the additional property that it maintains rainbow cycles between the two LDCs. By using the strong bounds of~\cite{ABSZZ23},\footnote{Note that this reduction crucially relies on the strong bound of $k \le O(\log{n} \log{\log{n}})$ by~\cite{ABSZZ23}. Indeed, if one instead uses the previous state-of-the-art results of~\cite{JS22, KLLT22} on rainbow cycles of $k \le O(\log^2{n})$, then this reduction would fail to yield any non-trivial bound.} we can therefore find a rainbow cycle in the $2$-LDC and revert it to a rainbow cycle in the $r$-LDC. From our general reduction, we can therefore deduce improved lower bounds of the form $k \le \widetilde{O}(n^{1 - 2/(r-1)})$ for binary linear $r$-LCCs for all odd $r \ge 5$, which were previously conjectured by~\cite{KM23} for all $r \ge 4$. This is the content of the following theorem.

\begin{theorem}
\label{thm:odd-lcc-lbs}
If $C$ is an $[n,k]$ binary linear $(r,\delta)$-LCC for odd $r \ge 5$, then $k \le O\left(\delta^{-2}n^{1 - \frac{2}{r-1}}\log^3{n}\right)$.
\end{theorem}

Note that the previously best known bound for binary linear $r$-LCCs for odd $r \ge 5$ (which also held for binary linear $r$-LDCs and even binary linear $(r+1)$-LDCs) was $\widetilde{O}(n^{1-2/(r+1)})$~\cite{KdW04, Woo07, HKMMS24}. We outline our general reduction and the proof of~\cref{thm:odd-lcc-lbs} in Section~\ref{sec:rainbow-ldc-higher-query}.

\paragraph{Follow-up questions.} 
Two salient follow-up questions to our work are removing the linearity assumption in~\cref{thm:main} and extending~\cref{thm:intro-cov-rad} to arbitrary finite fields. Since the statement of~\cref{thm:intro-cov-rad} crucially relies on considering rows of a generator matrix of the $3$-LCC, it makes it unclear how to remove the linearity assumption in~\cref{thm:main}. As for extending our main results to arbitrary finite fields, it is easy to extend~\cref{thm:intro-cov-rad} to finite fields of characteristic $2$ for a $\text{poly}(\abs{\F})$ loss in the upper bound on the size of the sum by considering the code defined in Appendix A of~\cite{KM23}. For finite fields of higher characteristic, the presence of negative signs presents a tricky situation for the application of the result of~\cite{ABSZZ23} in the proof of~\cref{thm:intro-cov-rad}. We leave it as an interesting open problem to extend~\cref{thm:intro-cov-rad} to linear $3$-LCCs over arbitrary finite fields. 

There is additionally the problem of extending~\cref{thm:odd-lcc-lbs} to all $r \ge 4$. In light of our proof method of~\cref{thm:odd-lcc-lbs}, it would seem that a cubic bound of $k \le \widetilde{O}(\sqrt[3]{n})$ for binary linear $[n,k]$ $4$-LCCs is reasonable to hope for by extending the cubic $3$-LDC lower bound of~\cite{AGKM23} to their analogous ``rainbow" version and applying our general reduction from binary linear $r$-LCC lower bounds to ``rainbow'' binary linear $(r-1)$-LDC lower bounds. However, the $3$-LDC to $2$-LDC transformation in~\cite{AGKM23} creates new query sets by adding together the original query sets, which disrupts the correspondence of the colors between the $3$-LDC and the derived $2$-LDC. Nonetheless, it would still be interesting to show a cubic ``rainbow'' binary linear $3$-LDC lower bound using the techniques of~\cite{AGKM23}.

\subsection{Proof overview}
\label{subsec:proof-overview}

While our proof of~\cref{thm:main} is rather short (just 2 pages, and self-contained modulo the rainbow cycle bound), we will nonetheless present a proof overview of it to showcase its key ideas. Consider a $(3,\delta)$-LCC whose generator matrix has $v_1,v_2,\dots,v_n \in \F_2^k$ as rows. It is well known that any binary linear $(3,\delta)$-LCC has a collection of hypergraphs $\mc{H}_1, \ldots , \mc{H}_n$ over $[n]$ such that for each $i \in [n]$, the hypergraph $\mc{H}_i$ consists of at least $(\delta/3)n$ disjoint subsets of $[n]$ of size $3$ each such that for any hyperedge $\{a,b,c\} \in \mc{H}_i$, we have that $v_i = v_a + v_b + v_c$ (see~\cref{subsec:lccs}). For simplicity, suppose that $\delta \ge \Omega(1)$ to ignore any $\delta$ dependencies. Our goal is to show that every $x \in \F_2^k$ can be represented as the sum of at most $B$ vectors in $\{v_1, \ldots, v_n\}$ for some $B \coloneqq \Theta(\log{n}\log\log{n})$.

Since the $v_i$'s span $\F_2^k$, $x$ can be written as the sum of at most $k$ of the $v_i$'s. Fix any such sum. Our proof proceeds in an iterative fashion: whenever the current representation of $x$ as a sum of the $v_i$'s is longer than $B$, we will exploit the many local checks expressing each $v_i$ as the sum of many disjoint $3$-tuples of other $v_j$'s to produce a shorter representation of $x$. Applying this compression iteratively yields the desired conclusion.

Now, consider an arbitrary linear combination $\sum_{t \in T}{v_t}$ with $\abs{T} > B$. For any $t \in T$, we can locally modify $\sum_{t \in T}{v_t}$ by applying the substitution $v_t = v_a + v_b + v_c$ for any $\{a,b,c\} \in \mathcal{H}_t$.  This will increase the length of the sum by (at most) $2$, which defeats our initial goal. Nonetheless, since $\abs{\mc{H}_t} \ge \Omega(n)$ for each $t \in T$, the abundance of choices for the triple $\{a,b,c\} \in \mc{H}_t$ presents a possibility for producing cancellations between substituted sums of triples of vectors.

The simplest form of such a cancellation between two substitutions goes as follows: consider any two distinct indices $t_1, t_2 \in T$ such that there are triples $\{a_1, b_1, c_1\} \in \mathcal{H}_{t_1}$ and $\{a_2, b_2, c_2\} \in \mathcal{H}_{t_2}$ satisfying $c_1 = b_2$. Since each hypergraph is a matching of size $\Omega(n)$, such triples do occur whenever $\abs{T} = \omega(1)$. Now, by applying the substitutions $v_{t_1} = v_{a_1} + v_{b_1} + v_{c_1}$ and $v_{t_2} = v_{a_2} + v_{b_2} + v_{c_2}$ in $\sum_{t \in T}{v_t}$, we obtain a new sum of length at most $\abs{T} + 2 \cdot 2  - 2 \cdot 1 = \abs{T}+2$ due to $v_{c_1}$ and $v_{b_2}$ canceling each other out.

We can further generalize this form of cancellation to multiple indices as follows: given distinct indices $t_1, \ldots , t_m \in T$ such that there exists a ``path" of hyperedges $E_s \coloneqq \{a_{E_s}, b_{E_s}, c_{E_s}\} \in \mathcal{H}_{t_s}$ for $s \in [m]$ satisfying $c_{E_s} = b_{E_{s+1}}$ for each $s \in [m-1]$, we can apply the substitutions $v_{t_s} = v_{a_{E_s}} + v_{b_{E_s}} + v_{c_{E_s}}$ for each $s \in [m]$ to the sum $\sum_{t \in T}{v_t}$ and obtain a new sum of length at most $\abs{T} + 2m - 2(m-1) = \abs{T}+2$ due to $v_{c_{E_s}}$ and $v_{b_{E_{s+1}}}$ canceling each other out for each $s \in [m-1]$. Thus the length of the new sum hardly deviates from the length of the original sum. Furthermore, by a simple counting argument, one can show that there are such ``paths" of length $m = \Omega(\abs{T})$. However, the length of this new sum is not smaller or even equal to the length of the original sum.

Now, notice that if we had $c_{E_m} = b_{E_1}$ (i.e., the path `loops back'), then the length of the new sum will now be at most $\abs{T}$. This does not reduce the length of the original sum $\sum_{t \in T}{v_t}$, but it does `shift' it to a new sum. In the sequel, we will exploit such `shifts' to produce a new sum of smaller length. For now, let us consider the feasibility of having $c_{E_m} = b_{E_1}$.

To do so, we will cast our problem in the language of \emph{properly edge-colored graphs} and \emph{rainbow cycles}. Indeed, consider the edge-colored graph $G_T$ with vertices $[n]$ and edges $\{b, c\}$ for $\{a,b,c\} \in \mathcal{H}_t$ (dropping an arbitrary vertex in each triple) forming the $t$'th color class of edges in $G_T$ for each $t \in T$. As $\mathcal{H}_t$ is a matching, $G_T$ will therefore be a properly edge-colored graph. In this viewpoint, the `path' of hyperedges $E_1, \ldots , E_m$ is in fact a \emph{rainbow path} in $G_T$ with edge colors $t_1, \ldots , t_m$ in that order. To have $c_{E_m} = b_{E_1}$, we need this rainbow path to be a \emph{rainbow cycle}. Since the average degree of $G_T$ equals $\Omega(|T|) = \Omega(B) = \Omega(\log n \log \log n)$, we can therefore conclude the existence of a rainbow cycle in $G_T$ by the recent breakthrough result of \cite{ABSZZ23}. This rainbow cycle gives an alternate representation $\sum_{t \in T'}{v_t}$ that equals $\sum_{t \in T}{v_t}$ with $\abs{T'} \le \abs{T}$. Call such an $T'$ a ``shift" of $T$. Since the hypergraphs $\{\mc{H}_t\}_{t \in T}$ are matchings of size $\Omega(n)$, we can in fact extract more from this argument. Specifically, by a more careful selection of the edges of $G_T$, we can show that the collection of all ``shifts" $T'$ of $T$ cover $\Omega(n)$ of the indices in $[n]$. This is the content of~\cref{lem:shift-representation}.

This now suffices for an actual compression of a somewhat larger sum. Suppose $x = \sum_{i \in I}{v_i}$ for $\abs{I} >  p \cdot (B+1)$ for some large enough constant $p$ (which will depend on $\delta$). Splitting the sum into $p$ disjoint parts $T_1, T_2,\dots, T_p$, each with more than $B$ terms, the constant fraction `coverage' of $[n]$ by the ``shifts" of each set $T_\ell$ means (by some simple pigeonholing) that we can find two distinct indices $\ell_1, \ell_2 \in [p]$ and ``shifts" $T_{\ell_1}'$ and $T_{\ell_2}'$ that intersect. By replacing the sets $T_{\ell_1}$ and $T_{\ell_2}$ with their respective ``shifts," we end up with a representation of $x$ with at most $\abs{I}-2$ of the $v_i$'s, which concludes our iterative compression argument. See Figure~\ref{fig:cancellation} for an illustration.

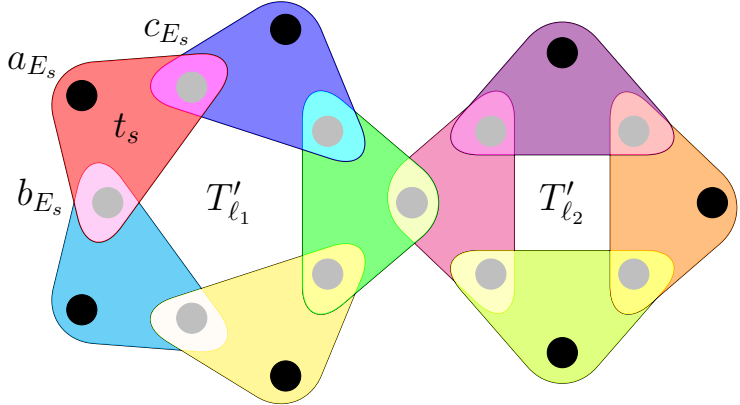
\begin{figure}
\begin{center}
\begin{tikzpicture}[scale=1.9]
\def\nodesize{3}
\def\shadelevel{50}
\def\activenodecolor{black}
\def\deadnodecolor{lightgray}

\def\pshift{0}
\def\p{1/1.17557}
\def\ph{1.5}
\def\ptri{1.9}
\def\proundedness{7}

\def\sshift{2.33}
\def\s{1}
\def\sh{2.1}
\def\stri{1.9}
\def\sroundedness{7}

\begin{scope}[blend group=overlay]


\draw [rounded corners= \proundedness mm,fill=red!\shadelevel] (\pshift + 0.3333*-0.809*\p*\ph + 0.3333*-1*\p + 0.3333*-0.309*\p + 0.6667*-0.809*\p*\ph*\ptri + -0.3333*-1*\p*\ptri + -0.3333*-0.309*\p*\ptri, 0.3333*0.5878*\p*\ph + 0.3333*0*\p + 0.3333*0.951*\p + 0.6667*0.5878*\p*\ph*\ptri + -0.3333*0*\p*\ptri + -0.3333*0.951*\p*\ptri)--(\pshift + 0.3333*-0.809*\p*\ph + 0.3333*-1*\p + 0.3333*-0.309*\p + -0.3333*-0.809*\p*\ph*\ptri + 0.6667*-1*\p*\ptri + -0.3333*-0.309*\p*\ptri, 0.3333*0.5878*\p*\ph + 0.3333*0*\p + 0.3333*0.951*\p + -0.3333*0.5878*\p*\ph*\ptri + 0.6667*0*\p*\ptri + -0.3333*0.951*\p*\ptri)--(\pshift + 0.3333*-0.809*\p*\ph + 0.3333*-1*\p + 0.3333*-0.309*\p + -0.3333*-0.809*\p*\ph*\ptri + -0.3333*-1*\p*\ptri + 0.6667*-0.309*\p*\ptri, 0.3333*0.5878*\p*\ph + 0.3333*0*\p + 0.3333*0.951*\p + -0.3333*0.5878*\p*\ph*\ptri + -0.3333*0*\p*\ptri + 0.6667*0.951*\p*\ptri)--cycle;

\draw [rounded corners= \proundedness mm,fill=blue!\shadelevel] (\pshift + 0.3333*0.309*\p*\ph + 0.3333*-0.309*\p + 0.3333*0.809*\p + 0.6667*0.309*\p*\ph*\ptri + -0.3333*-0.309*\p*\ptri + -0.3333*0.809*\p*\ptri, 0.3333*0.951*\p*\ph + 0.3333*0.951*\p + 0.3333*0.5878*\p + 0.6667*0.951*\p*\ph*\ptri + -0.3333*0.951*\p*\ptri + -0.3333*0.5878*\p*\ptri)--(\pshift + 0.3333*0.309*\p*\ph + 0.3333*-0.309*\p + 0.3333*0.809*\p + -0.3333*0.309*\p*\ph*\ptri + 0.6667*-0.309*\p*\ptri + -0.3333*0.809*\p*\ptri, 0.3333*0.951*\p*\ph + 0.3333*0.951*\p + 0.3333*0.5878*\p + -0.3333*0.951*\p*\ph*\ptri + 0.6667*0.951*\p*\ptri + -0.3333*0.5878*\p*\ptri)--(\pshift + 0.3333*0.309*\p*\ph + 0.3333*-0.309*\p + 0.3333*0.809*\p + -0.3333*0.309*\p*\ph*\ptri + -0.3333*-0.309*\p*\ptri + 0.6667*0.809*\p*\ptri, 0.3333*0.951*\p*\ph + 0.3333*0.5878*\p + 0.3333*0.951*\p + -0.3333*0.951*\p*\ph*\ptri + -0.3333*0.951*\p*\ptri + 0.6667*0.5878*\p*\ptri)--cycle;

\draw [rounded corners= \proundedness mm,fill=green!\shadelevel] (\pshift + 0.3333*1*\p*\ph + 0.3333*0.809*\p + 0.3333*0.809*\p + 0.6667*1*\p*\ph*\ptri + -0.3333*0.809*\p*\ptri + -0.3333*0.809*\p*\ptri, 0.3333*0*\p*\ph + 0.3333*0.5878*\p + 0.3333*-0.5878*\p + 0.6667*0*\p*\ph*\ptri + -0.3333*0.5878*\p*\ptri + -0.3333*-0.5878*\p*\ptri)--(\pshift + 0.3333*1*\p*\ph + 0.3333*0.809*\p + 0.3333*0.809*\p + -0.3333*1*\p*\ph*\ptri + 0.6667*0.809*\p*\ptri + -0.3333*0.809*\p*\ptri, 0.3333*0*\p*\ph + 0.3333*0.5878*\p + 0.3333*-0.5878*\p + -0.3333*0*\p*\ph*\ptri + 0.6667*0.5878*\p*\ptri + -0.3333*-0.5878*\p*\ptri)--(\pshift + 0.3333*1*\p*\ph + 0.3333*0.809*\p + 0.3333*0.809*\p + -0.3333*1*\p*\ph*\ptri + -0.3333*0.809*\p*\ptri + 0.6667*0.809*\p*\ptri, 0.3333*0*\p*\ph + 0.3333*0.5878*\p + 0.3333*-0.5878*\p + -0.3333*0*\p*\ph*\ptri + -0.3333*0.5878*\p*\ptri + 0.6667*-0.5878*\p*\ptri)--cycle;

\draw [rounded corners= \proundedness mm,fill=yellow!\shadelevel] (\pshift + 0.3333*0.309*\p*\ph + 0.3333*0.809*\p + 0.3333*-0.309*\p + 0.6667*0.309*\p*\ph*\ptri + -0.3333*0.809*\p*\ptri + -0.3333*-0.309*\p*\ptri, 0.3333*-0.951*\p*\ph + 0.3333*-0.5878*\p + 0.3333*-0.951*\p + 0.6667*-0.951*\p*\ph*\ptri + -0.3333*-0.5878*\p*\ptri + -0.3333*-0.951*\p*\ptri)--(\pshift + 0.3333*0.309*\p*\ph + 0.3333*0.809*\p + 0.3333*-0.309*\p + -0.3333*0.309*\p*\ph*\ptri + 0.6667*0.809*\p*\ptri + -0.3333*-0.309*\p*\ptri, 0.3333*-0.951*\p*\ph + 0.3333*-0.5878*\p + 0.3333*-0.951*\p + -0.3333*-0.951*\p*\ph*\ptri + 0.6667*-0.5878*\p*\ptri + -0.3333*-0.951*\p*\ptri)--(\pshift + 0.3333*0.309*\p*\ph + 0.3333*0.809*\p + 0.3333*-0.309*\p + -0.3333*0.309*\p*\ph*\ptri + -0.3333*0.809*\p*\ptri + 0.6667*-0.309*\p*\ptri, 0.3333*-0.951*\p*\ph + 0.3333*-0.5878*\p + 0.3333*-0.951*\p + -0.3333*-0.951*\p*\ph*\ptri + -0.3333*-0.5878*\p*\ptri + 0.6667*-0.951*\p*\ptri)--cycle;

\draw [rounded corners= \proundedness mm,fill=cyan!\shadelevel] (\pshift + 0.3333*-0.809*\p*\ph + 0.3333*-0.309*\p + 0.3333*-1*\p + 0.6667*-0.809*\p*\ph*\ptri + -0.3333*-0.309*\p*\ptri + -0.3333*-1*\p*\ptri, 0.3333*-0.5878*\p*\ph + 0.3333*-0.951*\p + 0.3333*0*\p + 0.6667*-0.5878*\p*\ph*\ptri + -0.3333*-0.951*\p*\ptri + -0.3333*0*\p*\ptri)--(\pshift + 0.3333*-0.809*\p*\ph + 0.3333*-0.309*\p + 0.3333*-1*\p + -0.3333*-0.809*\p*\ph*\ptri + 0.6667*-0.309*\p*\ptri + -0.3333*-1*\p*\ptri, 0.3333*-0.5878*\p*\ph + 0.3333*-0.951*\p + 0.3333*0*\p + -0.3333*-0.5878*\p*\ph*\ptri + 0.6667*-0.951*\p*\ptri + -0.3333*0*\p*\ptri)--(\pshift + 0.3333*-0.809*\p*\ph + 0.3333*-0.309*\p + 0.3333*-1*\p + -0.3333*-0.809*\p*\ph*\ptri + -0.3333*-0.309*\p*\ptri + 0.6667*-1*\p*\ptri, 0.3333*-0.5878*\p*\ph + 0.3333*-0.951*\p + 0.3333*0*\p + -0.3333*-0.5878*\p*\ph*\ptri + -0.3333*-0.951*\p*\ptri + 0.6667*0*\p*\ptri)--cycle;


\draw [rounded corners= \sroundedness mm,fill=violet!\shadelevel] (\sshift + 0.3333*0*\s*\sh + 0.3333*-0.5*\s + 0.3333*0.5*\s + 0.6667*0*\s*\sh*\stri + -0.3333*-0.5*\s*\stri + -0.3333*0.5*\s*\stri, 0.3333*0.5*\s*\sh + 0.3333*0.5*\s + 0.3333*0.5*\s + 0.6667*0.5*\s*\sh*\stri + -0.3333*0.5*\s*\stri + -0.3333*0.5*\p*\stri)--(\sshift + 0.3333*0*\s*\sh + 0.3333*-0.5*\s + 0.3333*0.5*\s + -0.3333*0*\s*\sh*\stri + 0.6667*-0.5*\s*\stri + -0.3333*0.5*\s*\stri, 0.3333*0.5*\s*\sh + 0.3333*0.5*\s + 0.3333*0.5*\s + -0.3333*0.5*\s*\sh*\stri + 0.6667*0.5*\s*\stri + -0.3333*0.5*\s*\stri)--(\sshift + 0.3333*0*\s*\sh + 0.3333*-0.5*\s + 0.3333*0.5*\s + -0.3333*0*\s*\sh*\stri + -0.3333*-0.5*\s*\stri + 0.6667*0.5*\s*\stri, 0.3333*0.5*\s*\sh + 0.3333*0.5*\s + 0.3333*0.5*\s + -0.3333*0.5*\s*\sh*\stri + -0.3333*0.5*\s*\stri + 0.6667*0.5*\s*\stri)--cycle;

\draw [rounded corners= \sroundedness mm,fill=orange!\shadelevel] (\sshift + 0.3333*0.5*\s*\sh + 0.3333*0.5*\s + 0.3333*0.5*\s + 0.6667*0.5*\s*\sh*\stri + -0.3333*0.5*\s*\stri + -0.3333*0.5*\s*\stri, 0.3333*0*\s*\sh + 0.3333*0.5*\s + 0.3333*-0.5*\s + 0.6667*0*\s*\sh*\stri + -0.3333*0.5*\s*\stri + -0.3333*-0.5*\p*\stri)--(\sshift + 0.3333*0.5*\s*\sh + 0.3333*0.5*\s + 0.3333*0.5*\s + -0.3333*0.5*\s*\sh*\stri + 0.6667*0.5*\s*\stri + -0.3333*0.5*\s*\stri, 0.3333*0*\s*\sh + 0.3333*0.5*\s + 0.3333*-0.5*\s + -0.3333*0*\s*\sh*\stri + 0.6667*0.5*\s*\stri + -0.3333*-0.5*\s*\stri)--(\sshift + 0.3333*0.5*\s*\sh + 0.3333*0.5*\s + 0.3333*0.5*\s + -0.3333*0.5*\s*\sh*\stri + -0.3333*0.5*\s*\stri + 0.6667*0.5*\s*\stri, 0.3333*0*\s*\sh + 0.3333*0.5*\s + 0.3333*-0.5*\s + -0.3333*0*\s*\sh*\stri + -0.3333*0.5*\s*\stri + 0.6667*-0.5*\s*\stri)--cycle;

\draw [rounded corners= \sroundedness mm,fill=lime!\shadelevel] (\sshift + 0.3333*0*\s*\sh + 0.3333*0.5*\s + 0.3333*-0.5*\s + 0.6667*0*\s*\sh*\stri + -0.3333*0.5*\s*\stri + -0.3333*-0.5*\s*\stri, 0.3333*-0.5*\s*\sh + 0.3333*-0.5*\s + 0.3333*-0.5*\s + 0.6667*-0.5*\s*\sh*\stri + -0.3333*-0.5*\s*\stri + -0.3333*-0.5*\p*\stri)--(\sshift + 0.3333*0*\s*\sh + 0.3333*0.5*\s + 0.3333*-0.5*\s + -0.3333*0*\s*\sh*\stri + 0.6667*0.5*\s*\stri + -0.3333*-0.5*\s*\stri, 0.3333*-0.5*\s*\sh + 0.3333*-0.5*\s + 0.3333*-0.5*\s + -0.3333*-0.5*\s*\sh*\stri + 0.6667*-0.5*\s*\stri + -0.3333*-0.5*\s*\stri)--(\sshift + 0.3333*0*\s*\sh + 0.3333*0.5*\s + 0.3333*-0.5*\s + -0.3333*0*\s*\sh*\stri + -0.3333*0.5*\s*\stri + 0.6667*-0.5*\s*\stri, 0.3333*-0.5*\s*\sh + 0.3333*-0.5*\s + 0.3333*-0.5*\s + -0.3333*-0.5*\s*\sh*\stri + -0.3333*-0.5*\s*\stri + 0.6667*-0.5*\s*\stri)--cycle;

\draw [rounded corners= \sroundedness mm,fill=magenta!\shadelevel] (\sshift + 0.3333*-0.5*\s*\sh + 0.3333*-0.5*\s + 0.3333*-0.5*\s + 0.6667*-0.5*\s*\sh*\stri + -0.3333*-0.5*\s*\stri + -0.3333*-0.5*\s*\stri, 0.3333*0*\s*\sh + 0.3333*-0.5*\s + 0.3333*0.5*\s + 0.6667*0*\s*\sh*\stri + -0.3333*-0.5*\s*\stri + -0.3333*0.5*\p*\stri)--(\sshift + 0.3333*-0.5*\s*\sh + 0.3333*-0.5*\s + 0.3333*-0.5*\s + -0.3333*-0.5*\s*\sh*\stri + 0.6667*-0.5*\s*\stri + -0.3333*-0.5*\s*\stri, 0.3333*0*\s*\sh + 0.3333*-0.5*\s + 0.3333*0.5*\s + -0.3333*0*\s*\sh*\stri + 0.6667*-0.5*\s*\stri + -0.3333*0.5*\s*\stri)--(\sshift + 0.3333*-0.5*\s*\sh + 0.3333*-0.5*\s + 0.3333*-0.5*\s + -0.3333*-0.5*\s*\sh*\stri + -0.3333*-0.5*\s*\stri + 0.6667*-0.5*\s*\stri, 0.3333*0*\s*\sh + 0.3333*-0.5*\s + 0.3333*0.5*\s + -0.3333*0*\s*\sh*\stri + -0.3333*-0.5*\s*\stri + 0.6667*0.5*\s*\stri)--cycle;
\end{scope}

\filldraw[\deadnodecolor] (\pshift + -1*\p, 0*\p) circle (\nodesize pt);
\filldraw[\deadnodecolor] (\pshift + -0.309*\p, 0.951*\p) circle (\nodesize pt);
\filldraw[\deadnodecolor] (\pshift + 0.809*\p, 0.5878*\p) circle (\nodesize pt);
\filldraw[\deadnodecolor] (\pshift + 0.809*\p, -0.5878*\p) circle (\nodesize pt);
\filldraw[\deadnodecolor] (\pshift + -0.309*\p, -0.951*\p) circle (\nodesize pt);

\filldraw[\activenodecolor] (\pshift + -0.809*\p*\ph, 0.5878*\p*\ph) circle (\nodesize pt);
\filldraw[\activenodecolor] (\pshift + 0.309*\p*\ph, 0.951*\p*\ph) circle (\nodesize pt);
\filldraw[\deadnodecolor] (\pshift + 1*\p*\ph, 0*\p*\ph) circle (\nodesize pt);
\filldraw[\activenodecolor] (\pshift + 0.309*\p*\ph, -0.951*\p*\ph) circle (\nodesize pt);
\filldraw[\activenodecolor] (\pshift + -0.809*\p*\ph, -0.5878*\p*\ph) circle (\nodesize pt);

\filldraw[\deadnodecolor] (\sshift + -0.5*\s, 0.5*\s) circle (\nodesize pt);
\filldraw[\deadnodecolor] (\sshift + 0.5*\s, 0.5*\s) circle (\nodesize pt);
\filldraw[\deadnodecolor] (\sshift + 0.5*\s, -0.5*\s) circle (\nodesize pt);
\filldraw[\deadnodecolor] (\sshift + -0.5*\s, -0.5*\s) circle (\nodesize pt);

\filldraw[\activenodecolor] (\sshift + 0*\s*\sh, 0.5*\s*\sh) circle (\nodesize pt);
\filldraw[\activenodecolor] (\sshift + 0.5*\s*\sh, 0*\s*\sh) circle (\nodesize pt);
\filldraw[\activenodecolor] (\sshift + 0*\s*\sh, -0.5*\s*\sh) circle (\nodesize pt);
\filldraw[\deadnodecolor] (\sshift + -0.5*\s*\sh, 0*\s*\sh) circle (\nodesize pt);

\node at (\pshift + -0.809*\p*\ph - 0.40*\p, 0.5878*\p*\ph + 0.25*\p) {\Large $a_{E_s}$};
\node at (\pshift + -1*\p - 0.55*\p, 0*\p + 0.05*\p) {\Large $b_{E_s}$};
\node at (\pshift + -0.309*\p - 0.2*\p, 0.951*\p + 0.45*\p) {\Large $c_{E_s}$};
\node at (\pshift + 0.3333*-0.809*\p*\ph + 0.3333*-1*\p + 0.3333*-0.309*\p, 0.3333*0.5878*\p*\ph + 0.3333*0*\p + 0.3333*0.951*\p) {\Large $t_s$};
\node at (\pshift, 0) {\Large $T_{\ell_1}'$};
\node at (\sshift, 0) {\Large $T_{\ell_2}'$};

\end{tikzpicture}
\end{center}
\caption{This figure indicates the cancellations that occur in our proof of~\cref{thm:intro-cov-rad} via iterative refinement of the representation of an arbitrary vector $x \in \F_2^k$ as a sum more than $p(B+1) = \Omega(\log{n}\log\log{n})$ of the $v_i$'s. The nodes represent indices in $[n]$, with the gray nodes indicating `canceled' nodes in the sum $\sum_{i \in I}{v_i}$, while the black nodes represent the `active' nodes in the sum. The inner gray nodes in the pentagon and the square are cancellations resulting from~\cref{lem:shift-representation}. The cancellation of the one outer gray node in common is the result of picking a common node between two `shifts' $T_{\ell_1}'$ and $T_{\ell_2}'$ of the sets $T_{\ell_1}$ and $T_{\ell_2}$, which is key idea in the proof of~\cref{thm:intro-cov-rad} from~\cref{lem:shift-representation}. In the figure, a sum of $9$ terms (the indices $t_s$ corresponding to each of the $9$ colors) is compressed into a sum of $7$ terms (the black nodes).}
\label{fig:cancellation}
\end{figure}

\paragraph{Proof comparison to~\cite{KM23, Yan24}.}
One salient common feature in our work and the works of~\cite{KM23, Yan24} is the \emph{chaining} of local checks. However, our implementation of chaining differs fundamentally from~\cite{KM23, Yan24}. In our work, we attempt to chain local checks to form a ``cyclical chain" (i.e., rainbow cycles) in order to establish~\cref{thm:intro-cov-rad}, resulting in a much shorter proof. On the other hand,~\cite{KM23, Yan24} consider a technically involved hypergraph decomposition of a superpolynomial number of chained local checks and then proceed to undertake a highly intricate ``row pruning" analysis to ensure that each hypergraph of chained local checks is ``spread-out." 
Admittedly, our proof relies on black-boxing known results from the rainbow cycle literature, some proofs of which are involved. Nonetheless, our proof offers modularity. In particular, any improvement to the result of~\cite{ABSZZ23} would immediately yield better lower bounds on binary linear $3$-LCC via our proof of~\cref{thm:intro-cov-rad}. On the other hand, improvements using the methods of~\cite{KM23, Yan24} would likely entail a re-do of their analysis (as was the case in~\cite{Yan24}).

\subsection{Organization}

In~\cref{sec:prelims}, we state the tools we need for locally correctable codes and edge-colored graphs. In~\cref{sec:binary-3-lcc-proof}, we present the proof of~\cref{thm:main} and~\cref{thm:intro-cov-rad}. In \cref{sec:rainbow-ldc-higher-query}, we define the notion of a ``rainbow'' LDC lower bound along with a generalization of~\cref{thm:main} and use them to prove~\cref{thm:odd-lcc-lbs}. Finally, in~\cref{sec:new-2-ldc-proof}, we present a covering radius upper bound for linear $2$-LDCs and discuss how to obtain the exponential blocklength lower bound from our proof.

\section{Preliminaries}
\label{sec:prelims}

Let $\N \coloneqq \{0, 1, 2, \ldots\}$, and let $\F_2 = \{0,1\}$ denote the finite field of size $2$. For any positive integer $n \in \Z_+$, we denote $[n] \coloneqq \{1, 2, \ldots , n\}$. For any set $X$ and number $k \in \N$, denote $\binom{X}{k} \coloneqq \{ A \mid A \subseteq X \ , \ \abs{A} = k\}$. Given two sets $A$ and $B$, let $A \oplus B \coloneqq (A \setminus B) \cup (B \setminus A)$ denote their symmetric difference. Given a vector $x \in \F_2^k$, let $\text{wt}(x)$ denote its Hamming weight (i.e., number of nonzero entries). For any two vectors $x,y \in \F_2^n$, let $d(x,y)$ denote their Hamming distance (i.e., the number of entries that they differ on). We will consider \emph{multi-sets} in this work, which are simply sets that allow elements to repeat. For any multi-set $A$, the cardinality of $A$, denoted $\abs{A}$, is the number of elements in $A$ (including repeated elements).

A \emph{hypergraph} is simply a collection of sets $\mathcal{H} \subseteq 2^{[n]}$. We call the sets in the hypergraph \emph{hyperedges} For any $\ell \in \Z_+$, we say that $\mathcal{H}$ is an $\ell$-uniform hypergraph if $\abs{A} = \ell$ for all $A \in \mathcal{H}$. We also say that $\mathcal{H}$ is a \emph{matching} if $A \cap B = \varnothing$ for all distinct $A,B \in \mathcal{H}$. If $\mathcal{H}$ is an $\ell$-uniform hypergraph and a matching, then we simply call it an \emph{$\ell$-uniform matching.}

\subsection{Locally correctable codes}
\label{subsec:lccs}

The following is the usual definition of a linear 3-query locally correctable code $C$ as having a local decoder.

\begin{defn}[Binary Linear LCC, local decoder definition]
\label{def:lcc-local-dec}
Given a binary linear code $C \subseteq \F_2^n$, we say that it is a $(r, \delta)$-locally correctable code (abbreviated $(r,\delta)$-LCC) for $r \in \N$ and $\delta \in (0,1)$ if the following holds: for any received codeword $y \in \F_2^n$ there exists a randomized algorithm $\mc{D}^y$ with oracle access to $y$ that takes an index $i \in [n]$ as input and satisfies the following properties: (1) $\mc{D}^y(i)$ makes at most $r$ queries to $y$, and (2) if there exists a codeword $c \in C$ satisfying $d(x,c) \le \delta n$, then $\mc{D}^y(i)$ outputs $c_i$ with probability at least $2/3$.
\end{defn}

While \cref{def:lcc-local-dec} is the typical definition of LCCs, we will instead be working with a more combinatorial definition that is amenable to lower bounds.

\begin{defn}[Binary Linear LCC, combinatorial definition]
\label{def:lcc-combi}
Given a linear code $C$ with generator matrix $M \in \F_2^{n \times k}$ whose columns form a basis for $C$, let $v_i \in \F_q^k$ be the $i$'th row of $M$ for $i \in [n]$. The code $C$ is said to be a $(r, \delta)$-locally correctable code (abbreviated $(r, \delta)$-LCC) for $r \in \N$ and $\delta \in (0,1)$ if there exists $r$-uniform matchings $\mathcal{H}_1, \ldots , \mathcal{H}_n$ over $[n]$ such that $\abs{\mathcal{H}_i} \ge \delta n$ for all $i \in [n]$, and for any $i \in [n]$ and $\{a_1, \ldots , a_r\} \in \mathcal{H}_i$, we have that $v_i = \sum_{s=1}^r{v_{a_s}}$.
\end{defn}

It is well-known from standard reductions~\cite{KT00, Yek12, DSW14} that any code satisfying \cref{def:lcc-local-dec} also satisfies \cref{def:lcc-combi} for a multiplicative loss of $1/r$ in $\delta$. Therefore, without loss of generality. we will assume throughout the paper that the notion of a binary linear $(r,\delta)$-LCC refers to~\cref{def:lcc-combi} rather than~\cref{def:lcc-local-dec}.

\begin{remark}
The definition of a linear $(r, \delta)$-LCC in~\cref{def:lcc-combi} is invariant of the choice of generator matrix $M$ for the code $C$. Indeed, any generator matrix for $C$ is of the form  $MB$ for some invertible matrix $B \in \F_q^{k \times k}$. The rows of $MB$ are $B^\top v_i$ for $i \in [n]$. By linearity, it therefore follows that $B^\top v_i = \sum_{s=1}^r{B^\top v_{a_s}}$ for any $i \in [n]$ and $\{a_1, \ldots , a_r\} \in \mc{H}_i$.
\end{remark}

\subsection{Edge-colored graphs}
\label{subsec:edge-colored-graphs}

An undirected graph $G = (V,E)$ consists of a set $V$ and a multi-set $E \subseteq \binom{V}{2}$.\footnote{Note that $G$ does not necessarily have to be simple. That is, edges are allowed to repeat.} Given two edges $e_1, e_2 \in E$, we say that $e_1$ is \emph{incident} to $e_2$ if they share a common vertex. A subset of edges $E_0 \subseteq E$ is said to be a \emph{matching} if no two different edges in $E_0$ are incident to each other. Given a set of colors $T$, we say that a graph $G$ is \emph{edge-colored} if it has an associated function $c : E \to T$, which we call an edge coloring. For graphs with an associated edge coloring, we write them as $G = (V, E, c)$. Given a color $t \in T$, the \emph{color class} of $t$ of $G$ is the multi-set of edges $c^{-1}(t)$. We say that $c$ is a \emph{proper} edge coloring if any two different incident edges $e_1, e_2 \in E$ have different colors. Equivalently, $c$ is a proper edge coloring if $c^{-1}(t)$ is a matching for all $t \in T$. 

With all this terminology at hand, we can now define a rainbow cycle.

\begin{defn}[Rainbow Cycle]
Given an edge-colored graph $G = (V,E,c)$, a rainbow cycle is a tuple of vertices $(i_1, i_2, \ldots , i_\ell, i_{\ell+1} = i_1) \in V^\ell$ such that $\{i_j, i_{j+1}\} \in E$ for all $j \in [\ell]$ and the multi-set of edges $\{\{i_j, i_{j+1}\} : j \in [\ell]\}$ is each assigned a different color by $c$. 
\end{defn}

We will now rely on the following theorem of~\cite{ABSZZ23}. Note that when the graph is not simple, one can easily find a rainbow cycle of length $2$ in the graph (as it is properly edge-colored).

\begin{theorem}[\cite{ABSZZ23}, Theorem 1.1]
\label{thm:abszz}
There exists a universal constant $c_0 > 0$ such that the following holds: any properly edge-colored $n$-vertex graph $G$ with at least $c_0n\log{n}\log\log{n}$ edges contains a rainbow cycle.
\end{theorem}

\section{Proof of main 3-LCC result}
\label{sec:binary-3-lcc-proof}

Let $C$ be an $[n,k]$ binary linear $(3, \delta)$-LCC. Throughout this section, fix a generator matrix $M \in \F_2^{n \times k}$ for $C$ with row vectors $v_1, \ldots , v_n \in \F_2^k$ and associated $3$-uniform matchings $\mathcal{H}_1, \ldots , \mathcal{H}_n$ over $[n]$. Our main result for this section is the following theorem, which is just Theorem~\ref{thm:intro-cov-rad} restated.

\begin{theorem}
\label{thm:compression}
For any vector $x \in \F_2^k$, there exists a set of indices $I \subseteq [n]$ satisfying $x = \sum_{i \in I}{v_i}$ and  $\abs{I} \le O(\delta^{-2}\log{n}\log\log{n})$.
\end{theorem}

Indeed, from \cref{thm:compression}, our main result \cref{thm:main} immediately follows. 

\begin{proof}[Proof of~\cref{thm:main} from~\cref{thm:compression}]
By \cref{thm:compression}, for each $x \in \F_2^k$, we know of a set $I_x \subseteq [n]$ of size at most $O(\delta^{-2}\log{n}\log\log{n})$ satisfying $x = \sum_{i \in I_x}{v_i}$. Now, for distinct $x,y \in \F_2^k$, it follows from the definition of $I_x$ that $I_x \neq I_y$. Since $\abs{I_x} \le O(\delta^{-2}\log{n}\log\log{n})$, then there are at most $n^{O(\delta^{-2}\log{n}\log\log{n})}$ possibilities for any $I_x$. Thus $2^k \le n^{O(\delta^{-2}\log{n}\log\log{n})}$, from which we conclude that $k \le O(\delta^{-2}\log^2{n}\log\log{n})$.
\end{proof}

It therefore suffices to establish~\cref{thm:compression}. For that, we will rely on the following key lemma. 

\begin{lemma}
\label{lem:shift-representation}
Let $c_0$ be the absolute constant from~\cref{thm:abszz}.
For any set $T \subseteq [n]$ of size at least $2c_0\delta^{-1}\log{n}\log\log{n}$, let $W \subseteq [n]$ be the set of indices $j \in [n]$ such that there exists a multi-set $T'$ of indices in $[n]$ with $j \in T'$ satisfying $\abs{T'} \le \abs{T}$ and
\begin{equation*}
    \sum_{t \in T}{v_t} = \sum_{t \in T'}{v_t} \ .
\end{equation*}
Then $\abs{W} \ge (\delta/2)n$.
\end{lemma}

Indeed, assuming \cref{lem:shift-representation}, \cref{thm:compression} follows as argued below.

\begin{proof}[Proof of \cref{thm:compression} from \cref{lem:shift-representation}]
Let $I \subseteq [n]$ be a set of minimal cardinality satisfying $x = \sum_{i \in I}{v_i}$. Such a set exists as the vectors $v_1 , \ldots , v_n$ span $\F_2^k$ (as $M$ is full rank). Assume (for the sake of a contradiction) that $\abs{I} \ge 10c_0\delta^{-2}\log{n}\log\log{n}$. Randomly partition $I$ into $p \coloneqq \lceil 4/\delta \rceil$ sets $T_1 \ldots , T_p$ of equal size. Then $\abs{T_\ell} \ge 2 c_0 \delta^{-1}\log{n}\log\log{n}$ for all $\ell \in [p]$. Thus we can apply \cref{lem:shift-representation} to find sets $W_1, \ldots , W_p$ of size at least $(\delta/2)n$ each satisfying the property stated in~\cref{lem:shift-representation}. Observe that $\sum_{\ell=1}^p{\abs{W_\ell}} \ge (4/\delta) \cdot (\delta/2)n = 2n > n$. Thus, we can find distinct $\ell_1, \ell_2 \in [p]$ such that there is an index $j \in W_{\ell_1} \cap W_{\ell_2}$. Without loss of generality, say $(\ell_1, \ell_2) = (1,2)$. Then by \cref{lem:shift-representation}, we can find multi-sets $T_1', T_2' \subseteq [n]$ with $j \in T_1' \cap T_2'$ satisfying $\abs{T_1'} \le \abs{T_1}$, $\abs{T_2'} \le \abs{T_2}$, and
\begin{equation}
    \sum_{i \in T_1}{v_i} = \sum_{i \in T_1'}{v_i} \ , 
    \quad \text{as well as} \quad 
    \sum_{i \in T_2}{v_i} = \sum_{i \in T_2'}{v_i} \ . \label{eq:compression}
\end{equation}
Now, define the multi-set $I' \coloneqq (T_1' \setminus \{j\}) \cup (T_2' \setminus \{j\}) \cup \cup_{\ell=3}^p{T_\ell}$. From \eqref{eq:compression}, we find that
\begin{align*}
x = \sum_{i \in I}{v_i} &= \sum_{i \in T_1}{v_i} + \sum_{i \in T_2}{v_i} + \sum_{\ell=3}^p{\sum_{i \in T_\ell}{v_i}} \\
&= \sum_{i \in T_1'}{v_i} + \sum_{i \in T_2'}{v_i} + \sum_{\ell=3}^p{\sum_{i \in T_\ell}{v_i}} \\
&= \Bigl(v_j + \sum_{i \in T_1' \setminus \{j\}}{v_i}\Bigr) + \Bigl(v_j + \sum_{i \in T_2' \setminus \{j\}}{v_i}\Bigr) + \sum_{\ell=3}^p{\sum_{i \in T_\ell}{v_i}} \\
&= \sum_{i \in I'}{v_i} \ .
\end{align*}
Thus $x = \sum_{i \in I'}{v_i}$. On the other hand, since $\abs{T_1'} \le \abs{T_1}$ and $\abs{T_2'} \le \abs{T_2}$, then we find that
\begin{align*}
\abs{I'} = \abs{T_1' \setminus \{j\}} + \abs{T_2' \setminus \{j\}} + \sum_{\ell=3}^p{\abs{T_\ell}} \le (\abs{T_1}-1) + (\abs{T_2}-1) + \sum_{\ell=3}^p{\abs{T_\ell}} = \abs{I}-2 \ .
\end{align*}
This contradicts the minimality of $I$, which is what we wanted to show.
\end{proof}

We now turn to the proof of \cref{lem:shift-representation}. For this part, we introduce some notations. For any hyperedge $E \in \cup_{i=1}^k{\mathcal{H}_i}$, write $E = \{a_E, b_E, c_E\}$ for $a_E, b_E, c_E \in [n]$, and let $e_E \coloneqq \{b_E, c_E\}$.

\begin{proof}[Proof of \cref{lem:shift-representation}]
Assume (for the sake of a contradiction) that $\abs{W} < (\delta/2)n$. Consider the graph $G$ consisting of $[n]$ as vertices, $T$ as edge colors, and for each $t \in T$, the set $\{e_E : E \in \mathcal{H}_t, a_E \notin W\}$ as the edges of the color class $t$. Because $\{\mathcal{H}_t\}_{t \in T}$ are $3$-uniform matchings, any color class of edges in $G$ will form a matching of edges, meaning that $G$ is properly edge-colored. Furthermore, because $\{\mathcal{H}_t\}_{t \in T}$ are each of size at least $\delta n$, each color class has at least $\abs{\mathcal{H}_t} - \abs{W} > \delta n - (\delta/2)n = (\delta/2) n$ edges. Thus $G$ has at least $(\delta/2) n \cdot \abs{T} \ge c_0 n \log{n} \log\log{n}$ edges. 

By~\cref{thm:abszz}, there exists a positive integer $m \ge 2$, distinct indices $t_1, \ldots , t_m \in T$, and hyperedges $E_s \in \mathcal{H}_{t_s}$ for $s \in [m]$ such that the edges $(e_{E_1}, \ldots, e_{E_m})$ form a rainbow cycle in $G$. This implies that $\oplus_{s=1}^m{e_{E_s}} = \varnothing$. Now, define the set $T_0 \coloneqq T \setminus \{t_1, \ldots , t_m\}$. Then we have that
\begin{align*}
\sum_{t \in T}{v_t} = \sum_{s=1}^m{v_{t_s}} + \sum_{t \in T_0}{v_t} &= \sum_{s=1}^m{\left(v_{a_{E_s}} + v_{b_{E_s}} + v_{c_{E_s}}\right)} + \sum_{t \in T_0}{v_t} \\
&= \sum_{s=1}^m{\left(v_{b_{E_s}} + v_{c_{E_s}}\right)} + \sum_{s=1}^m{v_{a_{E_s}}} + \sum_{t \in T_0}{v_t} \\
&= \sum_{i \in \bigoplus_{s=1}^m{e_{E_s}}}{v_i} + \sum_{s=1}^m{v_{a_{E_s}}} + \sum_{t \in T_0}{v_t} \\
&= \sum_{s=1}^m{v_{a_{E_s}}} + \sum_{t \in T_0}{v_t} \ .
\end{align*}
Thus if we define the multi-set $T' \coloneqq T_0 \cup \{a_{E_1}, \ldots , a_{E_m}\}$, then we see that $\abs{T'} = \abs{T}$ and $\sum_{t \in T}{v_t} = \sum_{t \in T'}{v_t}$. However, since $e_{E_s}$ is an edge in $G$ for each $s \in [m]$, then from the definition of $G$, we see that $a_{E_s} \notin W$ for all $s \in [m]$. This yields a contradiction by the definitions of $W$ and $T'$.
\end{proof}

\section{Rainbow LDC bounds and higher query LCCs}
\label{sec:rainbow-ldc-higher-query}

In this section, we develop the notion of ``rainbow'' LDC lower bounds and use the direct sum transformation of~\cite{KdW04} and the result of~\cite{ABSZZ23} to prove~\cref{thm:odd-lcc-lbs}.

One salient feature of the proof of~\cref{thm:intro-cov-rad} is that it crucially relies on the results of~\cite{ABSZZ23} (\cref{thm:abszz}) regarding the existence of rainbow cycles in properly edge-colored graphs, which was only feasible due to the $3$-uniformity of the query sets. For higher query complexities, we remedy this obstacle by introducing a hypergraph generalization of~\cref{thm:abszz}, stated below.

\begin{defn}[Rainbow LDC Lower Bound]
\label{def:rainbow-ldc-lbs}
For $\delta > 0$ and $r, n \in \N$ with $r \ge 2$, let $k_{\text{\normalfont rainbow}}^{(r)}(\delta, n)$ be the smallest natural number such that the following holds: for any arbitrary $r$-matchings $\mc{H}_1, \ldots , \mc{H}_k$ over $[n]$ with $k \ge k_{\text{\normalfont rainbow}}^{(r)}(\delta, n)$ satisfying $\abs{\mc{H}_i} \ge \delta n$ for all $i \in [k]$, there exists a nonempty collection of hyperedges $\mc{E} \subseteq \cup_{i=1}^k{\mc{H}_i}$ such that $\bigoplus_{E \in \mc{E}}{E} = \varnothing$ and $\abs{\mc{E} \cap \mc{H}_i} \le 1$ for all $i \in [k]$.
\end{defn}

We dub~\cref{def:rainbow-ldc-lbs} as the \emph{rainbow LDC lower bound} problem. Our choice of naming comes from the fact that upper bounds on $k_{\text{rainbow}}^{(r)}(\delta, n)$ formally prove limitations for binary linear $r$-LDCs. This can be seen from the viewpoint of LDC lower bounds as finding ``odd even covers," formally shown in~\cite{HKMMS24}.

\begin{prop}{\cite[Lemma 2.7]{HKMMS24}}
\label{prop:ldc-lbs}
For $\delta > 0$ and $r,n \in \N$ with $r \ge 2$, let $k_{\text{\normalfont odd}}^{(r)}(\delta, n) \in \N$ be the smallest natural number such that the following holds: for any arbitrary $r$-matchings $\mc{H}_1, \ldots , \mc{H}_k$ over $[n]$ with $k \ge k_{\text{\normalfont odd}}^{(r)}(\delta, n)$ satisfying $\abs{\mc{H}_i} \ge \delta n$ for all $i \in [k]$, there exists a nonempty collection of hyperedges $\mc{E} \subseteq \cup_{i=1}^k{\mc{H}_i}$ such that $\bigoplus_{E \in \mc{E}}{E} = \varnothing$ and $\abs{\mc{E} \cap \mc{H}_i}$ is \emph{odd} for some $i \in [k]$. Then any binary linear $(r,\delta)$-LDC\footnote{See~\cref{def:2-ldc} for a formal definition of a linear $(r,\delta)$-LDC.} of block length $n$ has dimension less than $k_{\text{\normalfont odd}}^{(r)}(\delta, n)$.
\end{prop}

Note that $k_{\text{odd}}^{(r)}(\delta, n) \le k_{\text{rainbow}}^{(r)}(\delta, n)$ as the property in \cref{def:rainbow-ldc-lbs} implies the property in~\cref{prop:ldc-lbs}. Now, with~\cref{def:rainbow-ldc-lbs} at hand, we can state our generalization of~\cref{thm:main}.

\begin{theorem}
\label{thm:main-general}
Let $\delta \in (0,1)$ and $r, n \in \N$ with $r \ge 3$. Then for any $[n,k]$ binary linear $(r,\delta)$-LCC, we have that
\begin{equation*}
    k \le O(\delta^{-1}  \cdot \log{n} \cdot k_{\text{\normalfont rainbow}}^{(r-1)}(\delta/2, n)) \ .
\end{equation*}
\end{theorem}

The proof of~\cref{thm:main-general} follows almost identically the proof of~\cref{thm:main} in~\cref{sec:binary-3-lcc-proof}. Indeed, the main property we needed from the rainbow cycle we found via~\cref{thm:abszz} was that the symmetric difference of the edges was the empty set and that every color appeared at most once. Thus if we generalize properly edge-colored graphs to properly edge-colored $(r-1)$-uniform hypergraphs\footnote{We say that an edge coloring of a hypergraph $\mc{H}$ is \emph{proper} if for any distinct hyperedges $e_1, e_2 \in \mc{H}$ satisfying $e_1 \cap e_2 \neq \varnothing$, $e_1$ and $e_2$ are assigned different colors.} and use~\cref{def:rainbow-ldc-lbs} in place of~\cref{thm:abszz} in~\cref{sec:binary-3-lcc-proof}, the proof of~\cref{thm:main-general} would then follow. To avoid redundancy, we leave the full proof of~\cref{thm:main-general} as an exercise for the reader.

As for upper and lower bounds on $k_{\text{rainbow}}^{(r)}(\delta, n)$, we know for $r = 2$ that $k_{\text{rainbow}}^{(2)}(\delta, n) \ge \Omega(\log{n})$ by considering the canonical coloring of the edges of the hypercube. Furthermore, by~\cref{thm:abszz}, we also know that $k_{\text{rainbow}}^{(2)}(\delta, n) \le O(\delta^{-1} \log n \log \log n)$. Now, as for $r \ge 3$, it follows from considering random $r$-uniform matchings that $k_{\text{rainbow}}^{(r)}(\delta, n) \ge \Omega_{\delta}(n^{1-2/r})$~\cite{HKM24}, which is a much higher lower bound than the bound $k_{\text{odd}}^{(r)}(\delta, n) \ge \exp(\Omega_{\delta}((\log{\log{n}})^2))$ for $r \ge 3$ obtained from known constructions of binary linear $r$-LDCs~\cite{Yek08, Efr12}.

Now, for the remainder of this section, we will prove the following proposition.

\begin{prop}
\label{prop:even-rainbow-ubs}
For any even $r \ge 4$ and $\delta \in (0,1)$, we have $k_\text{\normalfont rainbow}^{(r)}(\delta, n) \le O(\delta^{-1}n^{1-2/r}\log^2{n})$.
\end{prop}

Note that by combining~\cref{prop:even-rainbow-ubs} and~\cref{thm:main-general}, we immediately deduce~\cref{thm:odd-lcc-lbs}. Thus it suffices for us to prove~\cref{prop:even-rainbow-ubs}.

\begin{proof}[Proof of~\cref{prop:even-rainbow-ubs}]
We proceed by applying the direct sum transformation of~\cite{KdW04} to produce an edge-colored graph from the $r$-uniform matchings. Then using a deletion process similar to what was done in~\cite{GKM22, HKM23, AGKM23}, we will delete a sub-constant fraction of the edges from the graph to produce a properly edge-colored graph. We then apply~\cref{thm:abszz} to obtain a rainbow cycle and thus recover a rainbow even cover from it. The formal details follow.

Let $c_0$ be the absolute constant from \cref{thm:abszz}. We will show that for every choice of $r$-uniform matchings $\mathcal{H}_1, \ldots , \mathcal{H}_k$ over $[n]$ with $k \ge 2^{r^2+1}c_0\delta^{-1}n^{1-2/r}\log^2{n}$ and $\abs{\mathcal{H}_i} \ge \delta n$ for all $i \in [k]$, there is a nonempty subset of hyperedges $\mathcal{E} \subseteq \cup_{i=1}^k{\mathcal{H}_i}$ satisfying $\oplus_{E \in \mathcal{E}}{E} = \varnothing$ and $\abs{\mathcal{E} \cap \mathcal{H}_i} \le 1$ for all $i \in [k]$. This will imply $k_{\text{\normalfont rainbow}}^{(r)}(\delta, n) \le 2^{r^2+1}c_0\delta^{-1}n^{1-2/r}\log^2{n} = O(\delta^{-1} n^{1-2/r} \log^2 n)$.

Define $\ell \coloneqq 4^{-r}n^{1 - 2/r}$ and $N \coloneqq \binom{n}{\ell}$. Consider an edge-colored (not necessarily simple) graph $G$ over $\binom{[n]}{\ell}$ where two vertices $A, B \in \binom{[n]}{\ell}$ share an edge of color $i \in [k]$ if and only if $A \oplus B \in \mathcal{H}_i$. Fix any index $i \in [k]$ and hyperedge $E \in \mathcal{H}_i$. Observe that the number of sets $A,B \in \binom{[n]}{\ell}$ satisfying $A \oplus B = E$ is
\begin{equation}
\label{eq:edge-count}
    \binom{r}{r/2}\binom{n-r}{\ell - r/2} \ge N \cdot \left(\frac{\ell}{n}\right)^{r/2} = N \cdot \frac{2^{-r^2}}{n} \ .
\end{equation}
Now, let us upper bound the number of edges $\{A, B\}$ in $G$ of color $i$ satisfying $A \oplus B = E$ such that one of $A$ or $B$ is incident to another edge in $G$ of color $i$. Consider a set $A' \in \binom{[n]}{\ell}$ different from $B$ such that $\{A,A'\}$ is an edge in $G$ of color $i$. Define $E' \coloneqq A \oplus A' \in \mathcal{H}_i$. Because $\abs{A} = \abs{B} = \abs{A'} = \ell$, then we deduce that $\abs{A \cap E} = \abs{B \cap E} = \abs{A \cap E'} = \abs{A' \cap E'} = r/2$. Furthermore, because $A' \neq B$ and $\mathcal{H}_i$ is a matching, we have $E \neq E'$ and hence $E \cap E' = \varnothing$. Thus we find that $\abs{A \cap (E \cup E')} = r$. Now, since $\mathcal{H}_i$ is an $r$-uniform matching, then $\abs{\mathcal{H}_i} \le n/r$. Thus there are at most $n/r$ choices for $E'$ and hence at most $n/r$ choices for $E \cup E'$. For each such choice, there are at most $\binom{2r}{r} \binom{n-2r}{\ell-r}$ choices for $A'$. By repeating the same argument for $B$, we therefore deduce that the number of such edges $\{A,B\}$ is at most
\begin{equation}
\label{eq:repeat-edge-ub}
\frac{2n}{r} \cdot \binom{2r}{r} \binom{n-2r}{\ell-r} \le n \cdot 4^r \cdot \left(\frac{\ell}{n}\right)^r \cdot \binom{n}{\ell} = N \cdot \frac{4^r(4^{-r}n^{1 - 2/r})^r}{n^{r-1}} \le N \cdot \frac{2^{-2r^2+2r}}{n} \ .
\end{equation}
Now, let $G'$ be the edge-colored subgraph of $G$ consisting of all edges in $G$ that are not incident to any other edge of the same color. By definition, it follows that $G'$ is properly edge-colored. Furthermore, by combining~\eqref{eq:edge-count} and~\eqref{eq:repeat-edge-ub}, we find that the number of edges in $G'$ is at least
\begin{align*}
\left(N \cdot \frac{2^{-r^2}}{n} - N \cdot \frac{2^{-2r^2+2r}}{n}\right)\sum_{i=1}^k{\abs{\mathcal{H}_i}} &\ge N \cdot \frac{2^{-r^2-1}}{n} \cdot k \cdot \delta n \\
&= N \cdot 2^{-r^2-1} \delta k \\
&\ge N \cdot 2^{-r^2-1} \delta (2^{r^2+1} \delta^{-1} c_0 n^{1-2/r} \log^2{n}) \\
&= c_0 N \cdot n^{1 - 2/r}\log{n} \cdot \log{n} \\
&\ge c_0 N \cdot \log{N}\cdot \log{\log{N}} \ .
\end{align*}
Thus by~\cref{thm:abszz}, we can find a rainbow cycle in $G'$. That is, there exists $m \in \N$ and distinct indices $i_1, \ldots , i_m \in [k]$ and sets $A_1, A_2, \ldots , A_m, A_{m+1} = A_1 \in \binom{[n]}{\ell}$ such that $A_s \oplus A_{s+1} \in \mathcal{H}_{i_s}$ for all $s \in [m]$. Now, define $E_s \coloneqq A_s \oplus A_{s+1} \in \mathcal{H}_{i_s}$ for each $s \in [m]$. Then we find that
\begin{equation*}
    \bigoplus_{s=1}^m{E_s} = \bigoplus_{s=1}^m{(A_s \oplus A_{s+1})} = \bigoplus_{s=1}^m{A_s} \oplus \bigoplus_{s=1}^m{A_s} = \varnothing \ .
\end{equation*}
Thus if we define the set $\mathcal{E} \coloneqq \{E_1, \ldots , E_m\}$, then we see that $\oplus_{E \in \mathcal{E}}{E} = \varnothing$ and $\abs{\mathcal{E} \cap \mathcal{H}_i} \le 1$ for all $i \in [k]$, which is what we wanted to show.
\end{proof}

\section{Acknowledgements}

We thank Peter Manohar and Pravesh Kothari for helpful discussions that led to~\cref{thm:odd-lcc-lbs} and for showing that a conjecture stated in an earlier version of this work would imply the existence of binary linear constant-query LDCs of blocklength polynomial in the message length.

\bibliographystyle{alpha}
\bibliography{3-lcc}

\appendix

\section{A new proof of the exponential linear 2-LDC lower bound}
\label{sec:new-2-ldc-proof}

In this appendix, we present a new proof of the well-known exponential lower bound for linear $2$-query locally decodable codes~\cite{KdW04, GKST06} \`{a} la~\cite{IS20}. We begin by stating the definition of a linear $2$-query LDC for general finite fields $\F_q$. Note that this is usually referred to as a linear $2$-LDC in normal form, but by known reductions~\cite{Yek12}, the existence of a linear $2$-LDC implies the existence of a linear $2$-LDC in normal form. In what follows, the vectors $e_1, \ldots , e_k \in \F_q^k$ denote the standard basis.

\begin{defn}[Linear LDC]
\label{def:2-ldc}
Given a generator matrix $M \in \F_q^{n \times k}$, let $v_i$ the $i$'th row of $M$ for $i \in [n]$. For $r \in \N$ and $\delta > 0$, we say that $M$ forms a $(r,\delta)$-locally decodable code (abbreviated $(r,\delta)$-LDC) if there exist $r$-uniform matchings $\mathcal{H}_1, \ldots , \mathcal{H}_k$ over $[n]$ such that $\abs{\mathcal{H}_i} \ge \delta n$ for all $i \in [k]$, and for any $i \in [k]$ and $E = \{a_1, \ldots , a_r\}\in \mathcal{H}_i$, there exist $\alpha_s^E \in \F_q \setminus \{0\}$ for $s \in [r]$ satisfying $e_i = \sum_{s=1}^r{\alpha_s^E v_{a_s}}$.
\end{defn}

\begin{remark}
While the LCC property (\cref{def:lcc-combi}) is a property of the code, the LDC property is a property of the generator matrix of the code and not an inherent property of the code. That is, a different choice of generator matrix for the same code would not necessarily fulfill \cref{def:2-ldc}.
\end{remark}

We now state the key result driving this section, which is the following \emph{weight contraction} lemma.

\begin{lemma}
\label{lem:contraction-2-ldc}
For any $x \in \F_q^k$, there exist $a_1,a_2 \in [n]$ and $\gamma_1, \gamma_2 \in \F_q \setminus \{0\}$ satisfying 
\begin{equation*}
    \text{\normalfont wt}(x + \gamma_1 v_{a_1} + \gamma_2 v_{a_2}) \le (1-2\delta/q)\text{\normalfont wt}(x) \ .
\end{equation*}
\end{lemma}

\begin{proof}
The proof proceeds via a ``path coupling" style argument on the Cayley graph $\text{Cay}(\F_q^k, \{\alpha v_i \; : \; \alpha \in \F_q \setminus \{0\}, i \in [n]\})$ but with the Hamming distance acting as the contracted distance. For $q = 2$, if we have $y,z \in \F_2^k$ with $y + z = e_i$, then by evolving $(y,z)$ to $(y +  v_a,y + v_b)$ where $a \in [n]$ is uniform, and $b$ is $a$'s matched vertex in $\mathcal{H}_i$ if it exists and $b=a$ otherwise, we can reduce the Hamming distance between $y$ and $z$ with probability $\Omega(\delta)$. For arbitrary $y,z \in \F_q^k$, we consider their shortest path in $\text{Cay}(\F_q^k, \{\alpha e_i \; : \; \alpha \in \F_q \setminus \{0\}, i \in [k]\})$ and couple the vertices of each pair of edges along that path accordingly. We now proceed with the formal argument for general $q$ below.

Let $S \coloneqq \text{supp}(x)$ and $w \coloneqq \abs{S}$. Write $S = \{i_1, \ldots , i_w\}$ and $x = \beta_1 e_{i_1} + \ldots + \beta_w e_{i_w}$ for $\beta_t \in \F_q \setminus \{0\}$. Consider a uniformly randomly and independently chosen $\bs{\gamma}_0 \in \F_q \setminus \{0\}$ and $\mb{a}_0 \in [n]$. For each $t \in [w]$, define $\bs{\gamma}_t \in \F_q \setminus \{0\}$ and $\mb{a}_t \in [n]$ as
\begin{equation*}
(\bs{\gamma}_t, \mb{a}_t)
=
\begin{cases}
(-\bs{\gamma}_{t-1}(\alpha_{\mb{a}_{t-1}}^E)^{-1}\alpha_b^E, b) &\text{if $\exists \: b \in [n]$ such that $E \coloneqq \{\mb{a}_{t-1}, b\} \in \mc{H}_{i_t}$,} \\
(\bs{\gamma}_{t-1}, \mb{a}_{t-1}) &\text{otherwise.}
\end{cases}
\end{equation*}
Note that $b$ is well-defined in the first case as $\mc{H}_1, \ldots , \mc{H}_k$ are matchings. Furthermore, by a simple induction on $t$, it follows that $(\bs{\gamma}_t, \mb{a}_t)$ is uniformly random on $(\F_q \setminus \{0\}) \times [n]$ for all $t \in \{0, 1, \ldots , w\}$. Now, for each $t \in [w]$, define 
\begin{equation*}
\bs{\beta}_t'
=
\begin{cases}
\beta_t + \bs{\gamma}_{t-1}(\alpha_{\mb{a}_{t-1}}^E)^{-1} &\text{if $\exists \: b \in [n]$ such that $E \coloneqq \{\mb{a}_{t-1}, b\} \in \mc{H}_{i_t}$,} \\
\beta_t &\text{otherwise.}
\end{cases}
\end{equation*}
Then from the definitions, it follows that
\begin{equation}
\label{eq:one-swap}
\bs{\gamma}_{t-1}v_{\mb{a}_{t-1}} + \beta_t e_{i_t} = \bs{\beta}_t'e_{i_t} + \bs{\gamma}_tv_{\mb{a}_t}
\end{equation}
for all $t \in [w]$. Thus by iteratively applying \eqref{eq:one-swap}, we deduce that
\begin{equation*}
\bs{\gamma}_0v_{\mb{a}_0} + x = \bs{\gamma}_0v_{\mb{a}_0} + \beta_1 e_{i_1} + \ldots + \beta_w e_{i_w} = \bs{\beta}_1' e_{i_1} + \ldots + \bs{\beta}_w' e_{i_w} + \bs{\gamma}_w v_{\mb{a}_w} \ .
\end{equation*}
Thus we find that
\begin{equation}
\label{eq:random-coeff}
   x + \bs{\gamma}_0v_{\mb{a}_0} - \bs{\gamma}_wv_{\mb{a}_w} = \bs{\beta}_1' e_{i_1} + \ldots + \bs{\beta}_w' e_{i_w} \ .
\end{equation}
Now, for each $t \in [w]$, let $\mc{E}_t$ denote the event that there exists $b \in [n]$ such that $\{\mb{a}_{t-1}, b\} \in \mc{H}_{i_t}$. Because $\mb{a}_{t-1}$ is uniformly random over $[n]$ and $\mc{H}_{i_t}$ is a matching of size at least $\delta n$, it therefore follows that $\Pr{\mc{E}_t} \ge 2\delta$. Furthermore, in the event that $\mc{E}_t$ occurs, $\bs{\beta}_t'$ will be uniformly random over $\F_q \setminus \{\beta_t\}$ as $\bs{\gamma}_{t-1}$ is uniformly random over $\F_q \setminus \{0\}$. This implies that $\condPr{\bs{\beta}_t' = 0}{\mc{E}_t} \ge 1/q$. Hence we find that $\Pr{\bs{\beta}_t' = 0} \ge \condPr{\bs{\beta}_t' = 0}{\mc{E}_t}\Pr{\mc{E}_t} \ge 2\delta/q$. Now, let $\mb{X}$ be the number of $\bs{\beta}_1', \ldots , \bs{\beta}_w'$ that are equal to zero. By linearity of expectation, we find that $\Ex{\mb{X}} \ge (2\delta/q)w$. Thus, there exist $\gamma_0, \gamma_w \in \F_q \setminus \{0\}$ and $a_0, a_w \in [n]$ such that $\mb{X} \ge (2\delta/q)w$. From \eqref{eq:random-coeff}, we find that
\begin{equation*}
\wt(x + \gamma_0v_{a_0} - \gamma_wv_{a_w}) = \wt(\bs{\beta}_1' e_{i_1} + \ldots + \bs{\beta}_w' e_{i_w}) \le w - \mb{X} \le (1 - 2\delta/q)w \ ,
\end{equation*}
which completes our proof.
\end{proof}

By iteratively applying the above lemma an appropriate number of times, one can immediately deduce the following. 

\begin{theorem}
\label{thm:compression-2-ldc}
Suppose that a generator matrix $M \in \F_q^{n \times k}$ with rows $v_1,v_2,\dots,v_n \in \F_q^k$ forms a $(2,\delta)$-LDC. Then, for some absolute constant $c > 0$, the following holds for every $x \in \F_q^k$:
\begin{itemize}
    \item There exists $I \subseteq [n]$ with $\abs{I} \le c q\delta^{-1}\log{k}$ such that $x$ is in the $\F_q$-span of $\{v_i\}_{i \in I}$
\item There exist $J \subseteq [n]$ with $\abs{J} \le c q\delta^{-1}$ and $y$ in the $\F_q$-span of $\{v_j\}_{j \in J}$ such that the Hamming distance between $x$ and $y$ is at most $k/4$.
\end{itemize}
\end{theorem}

\noindent The exponential lower bound for 2-LDC now follows by essentially a covering radius argument.

\begin{theorem}
    \label{thm:2-ldc-exp-lb}
    Let $M \in \F_q^{n \times k}$ be a generator matrix that forms a $(2,\delta)$-LDC. Then $k \le O_{q,\delta}(\log n)$.
\end{theorem}
\begin{proof}
Let $v_1,\dots,v_n \in \F_q^k$ be the $n$ rows of $M$, and $c$ be the absolute constant from \cref{thm:compression-2-ldc}.
Define $W \subseteq \F_q^k$ to be the set of vectors which are in the span of at most $c q \delta^{-1}$ vectors amongst the $v_i$'s. 
Clearly 
\begin{equation}
    \label{eq:W}|W| \le (q n)^{c q \delta^{-1}} \ . 
\end{equation}
Let $U \subseteq \F_q^k$ consist of all vectors within Hamming distance $k/4$ from some element of $W$. By \cref{thm:compression-2-ldc}, $U = \F_q^k$. On the other hand, 
\begin{equation}
\label{eq:covering}
    q^k =  |U| \le |W| \cdot q^{h_q(1/4) k} \ .
\end{equation} 
where $h_q(x) := x \log_q (q-1) -x \log_q x - (1-x) \log_q(1-x)$ is the $q$-ary entropy function.  Combining \eqref{eq:W} and \eqref{eq:covering}, we conclude that $(1-h_q(1/4))k \le c q \delta^{-1}\log_q(qn)$ so that $k \le O_{q,\delta}(\log n)$ as desired.
\end{proof}

\end{document}